%% file: main.tex
\documentclass[10pt, conference, letterpaper]{IEEEtran}
\usepackage{cite}
\usepackage{amsmath,amssymb,amsfonts}
\usepackage{graphicx}
\usepackage{textcomp}
\usepackage{xcolor}
\usepackage{verbatim}
\usepackage{amsthm}
\usepackage{algorithm}
\usepackage[noend]{algorithmic}
\usepackage{comment}
\usepackage{epstopdf}
\usepackage{verbatim}

\newtheorem{lemma}{Lemma}

\newtheorem{theorem}[lemma]{Theorem}
\newtheorem{prop}{Proposition}

\usepackage{amsmath}

\def\BibTeX{{\rm B\kern-.05em{\sc i\kern-.025em b}\kern-.08em
    T\kern-.1667em\lower.7ex\hbox{E}\kern-.125emX}}

\begin{document}
\input{sections/title}
\input{sections/authors}
\maketitle
\begin{abstract}
\input{sections/abstract}
\end{abstract}

\input{sections/intro-infocom}

\input{sections/preliminaries}
\input{sections/background}

\input{sections/scheme}

\input{sections/sampling}

\input{sections/multi_mle}
\input{sections/evaluation.tex}
\input{sections/conclusions.tex}
\noindent\textbf{Acknowledgment.} This work was supported in part by the National Science Foundation under Grant No. CNS-2007006 and by a seed gift from Dolby Laboratories.

\bibliographystyle{IEEEtran}
\bibliography{bib/odd_sketch,bib/set_difference_from_pbs}

\end{document}

%% file: sections/title.tex
\title{OddEEC: A New Sketch Technique for Error Estimating Coding\\
\thanks{}
}

\author{\IEEEauthorblockN{Anonymous Authors}}

%% file: sections/authors.tex
\author{\IEEEauthorblockN{Huayi Wang \qquad\qquad Jingfan Meng \qquad\qquad Jun Xu}
  \IEEEauthorblockA{
    \textit{Georgia Institute of Technology}, Atlanta, GA, USA \\
    huayiwang@gatech.edu \qquad  jfmeng@gatech.edu \qquad jx@cc.gatech.edu}

}

%% file: sections/abstract.tex
Error estimating coding (EEC) is a standard technique for estimating the number of
bit errors during packet transmission over wireless networks. In this paper, we
propose OddEEC, a novel EEC scheme. OddEEC is a nontrivial
adaptation of a data sketching technique named Odd Sketch to EEC, addressing
new challenges therein by its bit sampling technique and maximum likelihood
estimator.  Our experiments show that OddEEC overall achieves comparable estimation accuracy as
competing schemes such as gEEC and mEEC, with much smaller decoding complexity.


%% file: sections/intro-infocom.tex
\section{Introduction}
\label{sec:intro}

\subsection{Problem Statement}~\label{subsec:prob_stat}

The problem of (approximately) estimating the \emph{bit error rate} (BER) on
wireless transmission links using little (bandwidth) overhead has been studied
for more than a decade~\cite{sigcomm_ChenZZY10, hua_simpler_2012,
  huang_eec_2014,meec}. Specifically, suppose that a sender transmits a packet
$P$ of $l$ bits long to a receiver over an unreliable wireless link. During the
transmission, the packet length remains the same, but some bits in it may be
flipped. As a result, the received packet $P'$ may differ from $P$ by some (say
$m$) bits; or equivalently, we say the \emph{Hamming distance} between $P$ and
$P'$ is $m$. The BER $\theta$ of this transmission is defined as $m/l$.

It is significant to know, in real time, this BER, since BER is a key metric for
the current wireless channel condition and makes possible a host of advanced
packet processing capabilities such as packet re-scheduling~\cite{re_scheduling},
routing~\cite{routing_ber}, and carrier-selection schemes~\cite{carrier_select} that can improve the
(good) throughput of a wireless network in various ways. To this end, the
sender can append a small \emph{error estimating coding} (EEC) codeword to $P$,
which, combined with the received packet $P'$, helps the receiver get an
estimation of the BER $\theta$ (between $P$ and $P'$).  

We identify two design requirements for an EEC scheme.  The first requirement is that an EEC scheme should achieve \emph{high estimation accuracy} using as little bandwidth overhead (codeword length) as possible.  This requirement is the design focus of all existing EEC schemes.  
Specifically, since the seminal EEC scheme~\cite{sigcomm_ChenZZY10}, several
works~\cite{hua_simpler_2012,geec,huang_eec_2014,meec} have refined EEC for
higher coding efficiency, or better BER estimation accuracy using a smaller EEC
codeword.  Among them, gEEC~\cite{geec} and mEEC~\cite{meec} not only achieved
the state-of-the-art coding efficiency, but also studied the general principles
for optimizing the coding efficiency.

The second design requirement concerns the \emph{decoding complexity} of an EEC scheme:  This complexity should be 
extremely low so that the receiver can decode a received EEC codeword (to obtain the estimated BER) in a tiny amount of time, say within a 
microsecond or less.  Unlike the first criterion, the second requirement was not a focus of
any existing EEC scheme.  This ``overlook" was ok when these schemes were developed (more than a decade ago):  Back then, 
the wireless transmission rates were not high enough to trigger this requirement.  However, this requirement becomes nonnegotiable now 
with today's wireless transmission rates.  For example, under Wi-Fi 6 (IEEE 802.11ax), whose maximum transmission rate can reach 9.6 Gbps~\cite{wifi_speed}, a 1500-byte-long wireless frame takes about 1.25 microseconds to transmit, and all frame processing (including EEC decoding) 
should finish well within this time frame.  In fact, none of the existing EEC schemes can meet this stringent (decoding speed) requirement, as we will elaborate in Section~\ref{sec:eval_eec}.

\subsection{OddEEC and Its Contributions}\label{subsec:intro_oddeec}
In this work, we propose a new EEC scheme, called OddEEC, that could make a
difference:  OddEEC achieves higher coding efficiency than all
existing schemes while having a significantly lower decoding time complexity,
making it well suited for today's high-speed wireless networks.
OddEEC builds upon an existing {\it data
    sketching} technique called Odd sketch, yet adapting Odd sketch into an EEC
scheme poses several new challenges. In this paper, we propose
techniques that overcome these challenges and theoretically bridge the gap
between data sketching and EEC. 

Before describing the challenges in our adaptation, we first introduce the Odd
sketch~\cite{odd-sketch}. Odd sketch is a space-efficient solution for
estimating the Jaccard similarity between two very large sets $A$ and $B$ (say
each containing millions of elements) that are stored at two different sites
connected by a communication network. We denote this problem by SSE-J (set
similarity estimation in Jaccard).
The goal of SSE-J is to estimate their Jaccard similarity $J(A,B) \triangleq
  |A\cap B|/|A \cup B|$ with a communication cost that is much smaller than that
of transmitting the entire $A$ (or $B$) to the other site. Odd sketch can be
viewed as a (sketching) function $S(\cdot)$ that summarizes $A$ (and $B$) into
a \emph{sketch} $S(A)$ (and $S(B)$) that is much smaller than $A$ (and $B$,
respectively) in size, yet the receiver can accurately estimate $J(A,B)$ from
$S(A)$ and $S(B)$ only (which are much less costly to transmit than $A$ and
$B$).

The first challenge in adapting Odd sketch into an EEC scheme is the difference
in similarity metric: Odd sketch was designed for measuring set similarity (in
Jaccard), whereas EEC, as mentioned above, estimates the Hamming distance
between the transmitted packet $P$ and the received packet $P'$. At a closer look,
however, we observe that the Hamming distance is equivalent to another set similarity metric called \emph{symmetric difference
  cardinality} (SDC), and that the Odd sketches of packets, namely $S(P)$ and
$S(P')$, are very informative about the SDC 
between $P$ and $P'$ (that can be converted to sets as will be explained later).

Our change of set similarity metric, from Jaccard to SDC,
also inadvertently leads to a second challenge: the
dimension reduction technique that Odd sketch employs to handle dissimilar sets
(corresponding to high BERs in this EEC context), namely the Broder's
embedding~\cite{minhash}, works only for Jaccard but not for SDC. To
this end, we develop and analyze a bit sampling technique in place of Broder's
embedding so that our OddEEC can achieve high coding efficiency in the high-BER
region.

\begin{figure*}[htb]
  \begin{center}
	    \includegraphics[width=1.8\columnwidth]{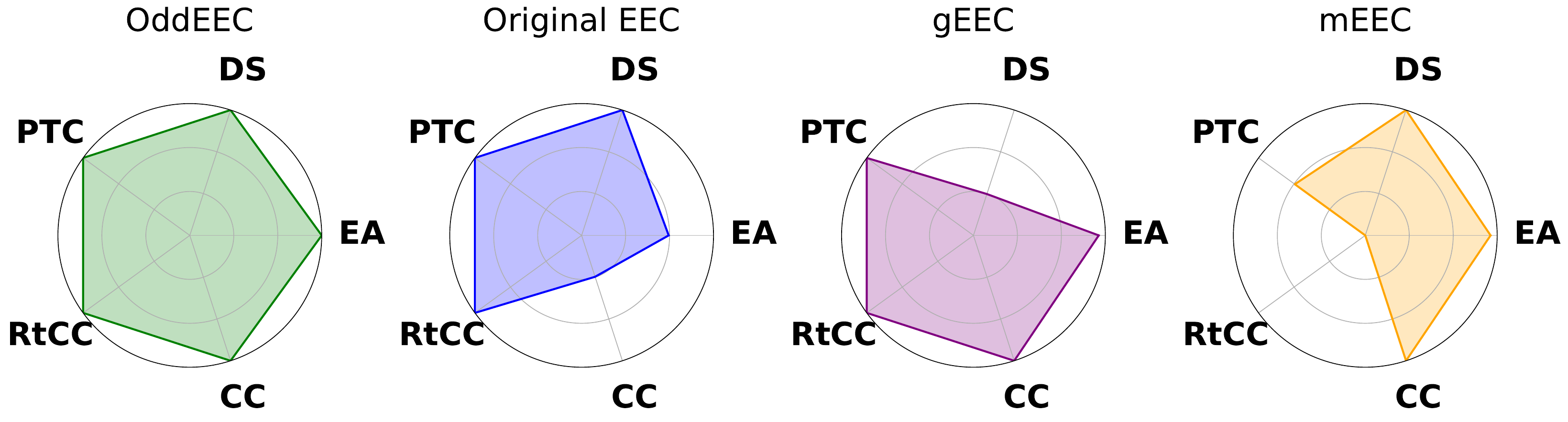}
	    \caption{ Comparison of OddEEC with three other EEC schemes in five aspects: precomputed table compactness (PTC), decoding speed (DS), estimation accuracy (EA), codeword compactness (CC), and robustness to codeword corruption (RtCC).}\label{fig:eval_trade}
	  \end{center}
\end{figure*}

The last, and the most consequential, challenge is caused by the corruption (by bit
errors) of EEC codewords. In a typical data sketching context, a sketch is immune
from bit errors (called the \emph{immunity} assumption in the sequel), but the
EEC codewords, which are transmitted along with packets, are subject to bit
errors under the same BER as the packets. The challenge is that it becomes much
harder to infer BER from a possibly corrupted EEC codeword at the receiver.
To address this challenge, we
develop a maximum likelihood estimator (MLE) that accurately ``accounts for''
possible bit errors in EEC codewords, in BER estimation. 
Building on our MLE framework, we develop a multi-resolution technique that uses multiple OddEEC subcodes parameterized
with different bit sampling rates.  By combining the estimation power and range of the subcodes, this technique allows for 
accurate BER estimation across a much broader range than a single-resolution scheme.

To summarize, we make the following three contributions on top of Odd sketch.
First, we adapt Odd sketch for estimating SDC instead of Jaccard, and develop a
bit sampling technique that is necessary for this change of similarity metric.
Second, we address the challenge of corrupted EEC codewords by developing an MLE,
which finalizes our adaptation of Odd sketch into OddEEC, a novel solution to
the BER estimation problem. Third, our evaluations demonstrate that, compared with the state-of-the-art EEC schemes,  OddEEC has comparable BER estimation accuracy with the same EEC codeword length and significantly lower decoding complexity.

%% file: sections/background.tex
\section{Background and related work}\label{sec:bg_eec}

We model the EEC coding problem as a data sketching problem. First known to the
networking community in early 2000s, data sketching has become a booming
networking research topic over the past
decade (e.g.,~\cite{one_sketch,elastic_sketch,Nitrosketch,bitsense}). A sketch
(function) $S(\cdot)$ summarizes a data vector say $v$ into a sketch (instance)
$S(v)$ that is much smaller than $v$ in size, but preserves certain important
features of $v$. In our EEC coding problem, we seek to design such an
$S(\cdot)$ that the Hamming distance (number of bit errors) between the
transmitted packet $P$ and the received packet $P'$ can be accurately estimated
from their respective sketches $S(P)$ and $S(P')$. With such an $S(\cdot)$
(which OddEEC is), the EEC codeword for, and sent along with, the transmitted
packet $P$ is precisely $S(P)$;  when the receiver receives $P'$, it can
compute $S(P')$, and compare it with $S(P)$ (received along with $P'$) to infer
the Hamming distance. We refer to this sketching problem as the Hamming
distance estimation (HDE) problem. There is a slight difference between the EEC
coding problem and HDE: In HDE, $S(P)$ is {\it immune} from bit errors, whereas in
this EEC coding problem, $S(P)$ is not. While the issue that $S(P)$ can itself
incur bit errors during transmission will certainly be addressed in our OddEEC
solution, we disregard it in our presentation below whenever possible, since this
(idealized) HDE formulation makes our presentation much easier and clearer.

\subsection{Related Work on EEC}\label{subsec:relate_eec}
The seminal paper of Chen et al.~\cite{sigcomm_ChenZZY10} first formulated the
BER estimation problem and proposed the error estimating coding (EEC) approach
to this problem, in which an EEC codeword is sent along with each packet (for
BER estimation). The original EEC scheme proposed therein is designed in a
similar way as an error-correction code~\cite{ecc_book}, although the former is
much smaller in size. In this scheme, an EEC codeword is comprised of
parity-check bits, each of which is the exclusive-or (XOR) of a group of bits
randomly sampled from the packet. Hence each parity bit, along with the group
of bits whose parity it ``checks'', corresponds to a parity equation.
As shown in~\cite{sigcomm_ChenZZY10}, the receiver can estimate this BER by
counting how many of these parity equations are broken after the packet
transmission.  However, the original EEC does not distinguish whether a packet data bit or an EEC bit is flipped in the parity equation. In contrast, the MLE estimator in OddEEC distinguishes these two cases using the likelihood values computed which we will describe in Section~\ref{subsec:mle}.

The second EEC scheme, called EToW~\cite{hua_simpler_2012} (enhanced
Tug-of-War) first considered the EEC coding problem from a data sketching
angle, and formulated it as the aforementioned HDE (Hamming Distance
Estimation) problem. Then, it was shown in~\cite{hua_simpler_2012} that HDE is
equivalent to the well-known $L_2$-norm estimation problem due to the following
fact: When $x$ and $y$ are binary strings (viewed as $l$-dimensional binary
vectors), their Hamming distance is identical to their Euclidean ($L_2$)
distance $||x-y||_2$. Hence, the second EEC coding scheme is an adaptation of a
well-established $L_2$-norm sketching solution called Tug-of-War (ToW). As
shown in~\cite{hua_simpler_2012}, it has achieved better coding efficiency than
the original EEC scheme.

The third EEC scheme combines the original EEC scheme and EToW into a
unified coding framework that they call generalized EEC (gEEC)~\cite{geec} in
the following sense: gEEC can be parameterized into both the original EEC
scheme and EToW. It also achieves much better coding efficiency than both. Specifically, gEEC (with 96-bit EEC codeword) achieves better BER estimation accuracy than the original EEC with 288-bit codeword as shown in~\cite{geec}. The weakness of the original EEC on the codeword compactness (CC) and estimation accuracy (EA) is captured in the following radar chart (the second subfigure from the left in Figure~\ref{fig:eval_trade}). 

The last EEC scheme is mEEC~\cite{meec}. In mEEC, a packet is divided into
equal-length chunks called super-blocks. 
mEEC maps each super-block independently to a color (code)
that is much shorter in length than the super-block, and the mEEC codeword is
the concatenation of the colors of all super-blocks. This mapping is designed
in such a way that the number of bit errors (and correspondingly the BER) that
occur to a super-block $B$ (and change it to $B'$) can be accurately inferred
from ``comparing'' $B'$ with the color of $B$, or the mEEC sub-codeword for $B$.
However, mEEC has a critical weakness: It makes the aforementioned
(unrealistic) assumption that each mEEC sub-codeword (like the color of $B$ here) is immune from bit errors, and it
does not have a ``contingency plan'' when this assumption fails. In contrast, our MLE estimator used in OddEEC 
is designed specifically for removing this assumption.

 In Figure~\ref{fig:eval_trade}, we use a radar chart to compare our OddEEC with the aforementioned three other EEC schemes qualitatively in five aspects: (1) precomputed table compactness (PTC); (2) decoding speed (DS); (3) estimation accuracy (EA); (4) codeword compactness (CC); (5) robustness to codeword corruption (RtCC). As shown in the figure, OddEEC outperforms or is comparable to all existing EEC schemes across all five metrics. Specifically, the original EEC has weaknesses in BER estimation accuracy and codeword compactness and mEEC is not robust to codeword corruption as described in the two paragraphs above.  We will elaborate why OddEEC outperforms gEEC in decoding speed and outperforms mEEC in precomputed table compactness in our evaluations in Section~\ref{sec:eval_eec}.  We do not include EToW in the radar chart because it was 
subsumed by its successor gEEC~\cite{geec}.

  \subsection{Related Work on other BER Estimation Methods}\label{subsec:ber_meth}
When packets are encoded using error-correction codes (with the intention of correcting all bit errors), they can provide BER estimation as a byproduct (even if they fail to correct all errors as implied by the CRC errors~\cite{crc_cisco}) as follows.
Error-correction codes like low-density parity-check (LDPC) codes
compute a soft BER estimate by first determining the log-likelihood ratio (LLR) for each received bit. Then, the posterior error probability for each bit is calculated from its LLR, and the average of these probabilities is used as the final BER estimate for the packet~\cite{ldpc_non_uniform, ldpc_lecture}.
	
However, the LDPC-based soft BER estimation technique has the following disadvantages compared to EEC schemes. First, it has much higher decoding complexity than EEC schemes.  As demonstrated in~\cite{eec_journal}, the LDPC decoding time is more than an order of magnitude slower than EEC decoding time.  Second, in order to determine the LLR for each bit, the codeword length of LDPC codes is usually very large. For example, as shown in the experiment results in~\cite{eec_journal}, LDPC codes need more than 1200 bits to reach a reasonable BER estimation accuracy for a 12,000-bit packet when BER is near 0.01. In contrast, the original EEC scheme only needs 288 bits to achieve better accuracy.  Furthermore, LDPC-based technique relies on knowing the channel model and its parameters, like the AWGN channel with known noise levels, so that it can generate accurate LLRs at the input of the decoder.  In contrast, all EEC schemes make no assumption except that each EEC bit has the same expected BER as the packet data bit, which can be achieved by uniformly randomly inserting EEC codeword bits within the packet as described in~\cite{sigcomm_ChenZZY10}.

\subsection{Related Work on Odd Sketch}
As explained earlier, Odd sketch was originally proposed for estimating the
Jaccard similarity between two large sets. Almost all
extensions (e.g.,~\cite{odd_sketch_ext}) and
applications (e.g.,~\cite{odd_sketch_graph,odd_sketch_apply}) of Odd sketch use it for
Jaccard similarity estimation. For example,~\cite{odd_sketch_ext} extends Odd
sketch for estimating multi-party Jaccard similarity (from the usual two-party
one).  One exception is
a dimension-reduction technique called Binary-to-Binary Dimension Reduction
(BDR)~\cite{dimen_reduct_oddsketch_2018}.
In BDR, a sparse high-dimensional binary vector is mapped to a dense
low-dimensional binary vector using the Odd sketch (function) $S(\cdot)$.
According to the aforementioned property of Odd sketch, two
binary vectors $\vec{x}$ and $\vec{y}$ (viewed as two sets) that are close in SDC (symmetric difference cardinality, which is equivalent to
Hamming distance as explained earlier) are mapped to $S(\vec{x})$ and $S(\vec{y})$ that
are also close in SDC. This approximate (Hamming) distance-preserving property
of BDR allows it to be used to speedup (approximate Hamming) distance
computations in computational tasks such as clustering. However, the BDR paper
did not study any estimator for the Hamming distance $\|\vec{x}, \vec{y}\|_H$,
likely because such estimators are not needed for the computational tasks that
BDR targets. In contrast, the focus of our OddEEC is to develop such estimators.

As we will show, for estimating SDC, Odd sketch is accurate only when the actual SDC value falls into a narrow range. 
A recent work, named GXBits~\cite{geo_oddsketch}, extends Odd sketch such that it works for a wider estimation range at a graceful loss of estimation 
accuracy.  It does so by changing the uniform hash function in Odd sketch to a geometric hash function (whose 
output has a geometric distribution).   
Similar to Odd sketch, a method of moments (MOM) estimator is derived for the SDC estimation.  However, the MOM estimator of GXBits requires solving a nonlinear optimization problem stemming from this geometric distribution.  This optimization problem is computationally intensive to solve:  It takes  
several to tens of milliseconds on a fast CPU~\cite{geo_oddsketch} using the fastest algorithms such as Newton-Raphson~\cite{Newton-Raphson}.  
While this high computation cost alone makes this MOM estimator, and hence the GXBits scheme, unsuitable for EEC, there is a more compelling disqualifier:
This MOM estimator relies on the aforementioned immunity assumption, and it is not known whether it can be extended to work without the 
immunity assumption;  in contrast, as explained earlier, an EEC scheme, and the estimator it uses, cannot rely on this assumption.

%% file: sections/scheme.tex
\section{Design of OddEEC}
\label{sec:algo}
In this section, we elaborate on our design of OddEEC, addressing the challenges
in adapting Odd sketch to EEC.
To begin with, in Section~\ref{subsec:overview}, we give an overview of Odd
sketch~\cite{odd-sketch}, on which our scheme is based.
Then, we describe how OddEEC addresses the three aforementioned challenges,
namely applying Odd sketch for Hamming (equivalent to symmetric difference cardinality, or
SDC) estimation in Section~\ref{subsec:oddeec}, our bit sampling scheme for
large packets and high BER in Section~\ref{subsec:samp}, and our maximum
likelihood estimator (MLE) for the \emph{without immunity} case (in which EEC
codewords may also be corrupted) in Section~\ref{subsec:mle}. Finally, we describe how we use multiple resolutions to extend OddEEC to large BER ranges in Section~\ref{subsec:multi_mle}.


\subsection{Overview of Odd Sketch}\label{subsec:overview}
As mentioned earlier, Odd sketch was proposed originally for set similarity estimation in
Jaccard similarity (SSE-J). An Odd sketch maps a large set $A$ to a relatively
tiny (usually $n\ll |A|$) binary array $S_A[1\cdots n]$, as follows. Denote
$[n]\triangleq \{1, 2, \cdots, n\}$. Let $h: U \rightarrow [n]$ ($U$ being the
universe) be a hash function, randomly seeded and fixed thereafter, that maps
every element $a\in U$ to a uniformly random value in $[n]$. Then, each bit of
$S_A[1 \cdots n]$ (say the $i^{th}$ bit $S_A[i]$ for $i = 1, 2, \cdots, n$) is
populated as follows. We count the number of elements $a\in A$ such that $h(a)
  = i$. If this number is {\it odd}, then $S_A[i]=1$;  if this number is {\it even}, then $S_A[i]=0$. In the
sequel, we drop the notation ``$[1\cdots n]$'' from $S_A[1\cdots n]$ whenever
the size $n$ of the Odd sketch is known from the context.

We now describe an important property of Odd sketch that makes it a suitable solution 
to SSE (in both Jaccard and SDC similarity metrics).  Let $S_B$ be the Odd sketch of another set $B$. The property is that,
$S_A\oplus S_B$, the bitwise-XOR of the two Odd sketches for sets $A$ and $B$ respectively, is
equal to the Odd sketch of their symmetric difference $A \triangle B$ as shown in~\cite{odd-sketch}. That is,
denoting the Odd sketch of $A \triangle B$ by $S_{A \triangle B}$, we have
$S_{A \triangle B} = S_A\oplus S_B$.

The symmetric difference cardinality (SDC), denoted by $m = |A \triangle B|$,
can be inferred from the Odd sketch $S_{A \triangle B}$ as follows. Denote by
$z$ the number of 1-bits in $S_{A \triangle B}$. It is clear that $z$ can be
modeled as the ``aftermath'' of the following ``$m$ balls into $n$ bins''
process: Throw the $m$ elements (balls) in $A \triangle B$ uniformly at random
(according to the hash function $h$) into the $n$ bits (bins) in $S_{A
      \triangle B}$. Then, $z$ is equal to the number of bins that end up with an odd
number of balls in it.
Given an observed $z$ (from $S_{A \triangle B}$), the following method of moments (MOM) estimator
for $m$ is given
in~\cite{odd-sketch}:
\begin{equation}\label{eq:estimator}
  \hat{m} = -\frac{n}{2} \ln(1-\frac{2z}{n})
\end{equation}
Using this $\hat{m}$ (estimated SDC), Odd sketch estimates the Jaccard
similarity between $A$ and $B$ by $\hat{J}(A, B) =(|A|+|B|-\hat{m})/(|A|+|B|+\hat{m})$.


The Odd sketch used in OddEEC is a slight modification of the original Odd sketch as follows~\cite{odd-sketch}.  
In the modified Odd sketch, each bin $i$ ($i = 1, 2, \cdots, n$) samples a subset $\eta_i$ from the set of $m$ balls as follows:  Each ball $j$
($j = 1, 2, \cdots, m$) is independently sampled into $\eta_i$ with probability $1/n$.    
Let $m_i=|\eta_i|$ ($i = 1, 2, \cdots, n$) be the number of balls sampled into bin $i$.  
This modification ensures that the random variables (RVs) $m_1$, $m_2$, $\cdots$, $m_n$ are (mutually) independent.
The original Odd sketch works both with and without this modification because its MOM estimator 
uses only the first moment, which is unbiased whether or not these RVs are independent.  
In contrast, OddEEC, which uses an MLE instead (of an MOM estimator) as explained in the second last paragraph of Section~\ref{subsec:intro_oddeec}, requires this modification since the MLE 
requires these RVs to be independent for the MLE to produce accurate estimates, especially when $n$ is small (less than 100) as in this EEC scenario.
Note that we use the term sampling here only to clearly and succinctly explain this modification:  There is no mathematical necessity to use this term here.
Rather, OddEEC involves a different, ``fully qualified'' sampling process that we
will describe in the following Section~\ref{subsec:samp}.  Hence, to avoid any confusion between these two processes, from this point on we will use this term exclusively 
for describing the ``fully qualified'' sampling process.

\subsection{OddEEC with Immunity}\label{subsec:oddeec}
From this subsection on, we describe our adaptation of Odd sketch to an EEC scheme.
As we mentioned earlier, given a packet $P$ (to be transmitted), the sender
appends an EEC codeword $S_P$ (actually the Odd sketch of $P$, as we will show
shortly) to it. After its transmission over the wireless link, $P$ may
encounter some bit errors and become $P'$. However, for ease of presentation,
for the moment we assume that the EEC codeword $S_P$ is \emph{immune} to bit
errors, or that the codeword $S_P$ is received intact (without any bit error). This \emph{with immunity}
assumption will be removed in our MLE scheme in Section~\ref{subsec:mle}.


As just explained, Odd sketch estimates the Jaccard similarity between two
sets, whereas in the EEC problem, the goal is to estimate the \emph{Hamming
  distance} between the transmitted packet $P$ and the received packet $P'$. The first
challenge in our adaptation, therefore, is to address this difference in
similarity metric. Fortunately, we observe that the Hamming distance in EEC is
equivalent to the following SDC metric when $P$ and $P'$ are viewed as two sets
as follows.
We define the set $P$ (with slight abuse in notation) as the set of bit
locations in the packet $P$ whose values are $1$ and define the set $P'$
similarly. In this way, the set of bit locations where errors occur is exactly the symmetric difference $P\triangle P'$, so
the Hamming distance between $P$ and $P'$ is exactly the SDC $|P\triangle P'|$.
The resulting SDC estimation problem, with immunity, can be readily solved by
Odd sketch as we mentioned above. 
In summary, our OddEEC, adapted from Odd sketch, works as follows:
\begin{itemize}
  \item At the sender, the EEC codeword is exactly $S_P$, the Odd sketch of the set $P$
        defined above.
  \item At the receiver, let $S_{P'}$ be the Odd sketch of set $P'$ (computed using the
        same hash function $h(\cdot)$ as the sender's) and $z$ be the number of 1-bits
        in $S_P\oplus S_{P'} = S_{P \triangle P'}$ (explained above). The estimated SDC
        (Hamming distance) $\hat{m}$ is calculated from $z$ by Eq.~\ref{eq:estimator}.
\end{itemize}

\input{sections/scheme_analysis}

%% file: sections/scheme_analysis.tex
The estimation accuracy of our OddEEC scheme can be analyzed as follows. 
To
begin with, we recall the following formula in Odd sketch~\cite{odd-sketch} for
the variance of $z$, the number of 1-bits in $S_P\oplus S_{P'}$ from which
$\hat{m}$ is estimated

\begin{align}
  \boldsymbol{\mathrm{Var}}[z] & = \frac{n^2}{4}((1-\frac{4}{n})^m-(1-\frac{2}{n})^{2m})+\frac{n}{4}(1-e^{-4m/n})\nonumber
  \\&\approx \frac{n}{4}(1-e^{-4m/n}-\frac{4m}{n} e^{-4m/n}).\label{eq:varz_approx}
\end{align}

We note that we do not consider the additional error introduced by our slight modification to Odd sketch (described at the end of Section~\ref{subsec:overview}) here, as its impact on the resulting approximation is negligible.

The following theorem gives an approximate formula for the variance of our MLE $\hat{m}$ (with immunity). This formula, after being extended in Section~\ref{subsec:samp}
for bit sampling, will guide our near-optimal bit sampling policy
later.





\begin{theorem}\label{the:var}
  $\boldsymbol{\mathrm{Var}}[\hat{m}] \approx \frac{n}{4}(e^{4m/n}-4m/n-1)$
\end{theorem}

Theorem~\ref{the:var} can be directly inferred from the combination of
Eq.~\ref{eq:varz_approx} and the following proposition.

\begin{prop}\label{prop:var}
  $\boldsymbol{\mathrm{Var}}[\hat{m}] \approx e^{4m/n} \boldsymbol{\mathrm{Var}}[z] $
\end{prop}

\begin{proof}
  Denoting the function $\log(1-\frac{2z}{n})$ as $f(z)$, we approximate the RHS
  of Eq.~\ref{eq:estimator} by the Taylor series expansion of $f(z)$ around
  $\alpha =      \boldsymbol{\mathrm{E}}[z]= \frac{n}{2}(1-e^{-2m/n})$.
    \begin{align}
    \hat{m} & = -\frac{n}{2} \times f(z) \nonumber                                                                         \\
            & = - \frac{n}{2} (f(\alpha) + (z-\alpha)f'(\alpha) + \frac{1}{2}(z-\alpha)^2 f''(\alpha) + \cdots) \nonumber \\
            & = - \frac{n}{2}(-\frac{2m}{n} - \frac{2}{n-2\alpha} (z-\alpha) + \cdots)  \nonumber                          \\
            & \approx m+\frac{n}{n-2\alpha}(z-\alpha) \label{eq:exp_approx}
  \end{align}

  As shown in~\cite{whang} for a similar but different estimator (designed for
   a similar but different sketch), truncating after two terms in Taylor
  series like in Eq.~\ref{eq:exp_approx} is a reasonable approximation with relative error generally
  less than 10\%. Based on Eq.~\ref{eq:exp_approx}, we arrive at
  Proposition~\ref{prop:var} as follows:
  \begin{align}
     & \boldsymbol{\mathrm{Var}}[\hat{m}] \approx \boldsymbol{\mathrm{Var}}[m+\frac{n}{n-2\alpha}(z-\alpha) ]  \nonumber \\
     & = (\frac{n}{n-2\alpha})^2  \boldsymbol{\mathrm{Var}}[z]  \nonumber                                                \\
     & = (\frac{n}{n \times e^{-2m/n}})^2 \boldsymbol{\mathrm{Var}}[z]  = e^{4m/n} \boldsymbol{\mathrm{Var}}[z]
  \end{align}
\end{proof}

%% file: sections/sampling.tex
\subsection{Bit Sampling and Analysis}\label{subsec:samp}
In this subsection, we describe and analyze our bit sampling technique in
OddEEC, which is essential for accurate BER estimation for long packets and
high BER values. To keep this analysis as simple as possible, we keep the
aforementioned immunity assumption (of EEC codewords).

\noindent
\textbf{The \emph{Saturation} Issue.} 
To motivate our bit sampling technique, we show a challenging scenario for
OddEEC, in which without bit sampling, the sketch (EEC codeword) would
``saturate'' and result in severe accuracy degradation. As we will describe
in~Section~\ref{sec:eval_eec}, the maximum target packet length of EEC is about $12,
  000$ bits and the maximum target error rate is about $0.01$. When these two
maximum parameter values are taken, each packet would have $m=12,000 \times
  0.01 = 120$ bit errors on average. Each OddEEC codeword, in this case, is a
binary array of $n=96$ bits (matching that in gEEC~\cite{geec}, the state of
the art).

The issue is that, without bit sampling, the value of $m$ is
way above the desired range for accurate estimation, for the following reason.
Recall that $\hat{m}$ is inferred from $z$ as in Eq.~\ref{eq:estimator}. As
shown in~\cite{odd-sketch}, $z$ is a binomial RV (random variable)
$\mathcal{B}(n,p)$ with $n$ trials and success rate
\begin{equation} \label{eq:prob_odd}
  p = (1-(1-2/n)^m)/2.
\end{equation}
When $m=120$ and $n=96$, we have $p=0.46$. As a result, the observed value of $z$, as a realization of the
above binomial RV $\mathcal{B}(96, 0.46)$, is greater than or equal to $n/2=48$ with roughly $1/4$ probability.
However, by Eq.~\ref{eq:estimator}, $\hat{m}$ is undefined when $z \ge n/2$,
which means that OddEEC would fail to produce any meaningful estimation every one out of four times!

More precisely,
when $m$ is large (compared to $n$), $p$ converges toward $1/2$ (from the
left), so $z$ would have a nontrivial probability to be close to or even
greater than $n/2$. Such realizations of $z$ would lead to estimates that are
highly erroneous or even meaningless, since $\hat{m}$, as a function of $z$ as
in Eq.~\ref{eq:estimator}, has a steep upward slope when $z$ approaches $n/2$
(from the left) and is undefined when $z \ge n/2$.

A natural solution to this issue is to reduce the ``number of balls $m$" in the
codewords, 
so that the sketch does not saturate.
In the original Odd sketch~\cite{odd-sketch}, this is done by the Broder's
embedding~\cite{minhash}, which compresses two large sets $A$ and $B$ into two
smaller ones $A'$ and $B'$, respectively. This embedding works for the original
Odd sketch (that estimates the Jaccard similarity), since it ensures $J(A, B)
  \approx J(A', B')$, or in other words roughly preserves the Jaccard similarity.  However, it does not work in OddEEC,
since $|A' \triangle B'|\ll |A \triangle B|$ (as $A'$ and $B'$ are much
smaller in size), or in other words it significantly distorts the Hamming distance.

\noindent
\textbf{The Bit Sampling Technique.}
Hence, we use the following bit sampling technique
instead (similarly to that in previous EEC works~\cite{hua_simpler_2012,geec})
that reduces $m$ to a known proportion of the Hamming distance between the transmitted packet $P$ and 
the received packet $P'$.  As mentioned in the last paragraph in Section~\ref{subsec:overview}, we simply sample bits uniformly without replacement at rate
$\beta$, from both $P$ and $P'$ consistently (i.e., at the same bit locations).  
We denote the sampled packets (or sets of sampled bits) as $P_{sam}$ and $P'_{sam}$ respectively.  With
bit sampling, the OddEEC codewords $S_P$ and $S_{P'}$ become Odd sketches of
$P_{sam}$ and $P'_{sam}$, respectively.  Nevertheless, we keep using their
original notations ($S_P$ and $S_{P'}$) in the sequel for simplicity.  This bit sampling operation can be performed in a computationally efficient manner
using SIMD instructions or hardware as we will elaborate in the end of Section~\ref{subsec:multi_mle}.

Since the Hamming distance between $P$ and $P'$ is roughly reduced by a factor of $1/\beta$ by
this bit sampling, our estimator $\hat{m}$ needs to be scaled accordingly as
follows
\begin{equation}\label{eq:samp_estimator}
  \hat{m} = \frac{\hat{\mu}}{\beta} =  -\frac{n}{2\beta} \log(1-\frac{2z}{n}).
\end{equation}
Note that here we use a different letter $\mu$ to denote the Hamming distance between $P_{sam}$ and $P'_{sam}$ (and $\hat{\mu}$ for its estimate computed from Eq.~\ref{eq:estimator}), so that it is not to be confused with the Hamming distance $m$ and its estimate $\hat{m}$ between unsampled packets $P$ and $P'$.


\noindent
\textbf{Variance Analysis.}
The following theorem extends Theorem~\ref{the:var} to the variance of
$\hat{m}$ with bit sampling, in which the first term corresponds to the RHS of
Theorem~\ref{the:var} with $m$ replaced by $\beta m$ (as an effect of bit
sampling), and the second term corresponds to the extra variance introduced by
the bit sampling process.
\begin{theorem}\label{theo:var_samp}
  $\boldsymbol{\mathrm{Var}}[\hat{m}] \approx \frac{n}{4 \beta^2}(e^{4\beta m/n}-4\beta m/n-1)+ \frac{m(1-\beta)}{\beta}$
\end{theorem}


\begin{proof}
  $\boldsymbol{\mathrm{Var}}[\hat{m}]$ can be divided into three parts as follows:
  \begin{align}
    \boldsymbol{\mathrm{Var}}[\hat{m}] & = \boldsymbol{\mathrm{E}}[(\frac{\hat{\mu}}{\beta}-m)^2] = \frac{1}{\beta^2} \boldsymbol{\mathrm{E}}[(\hat{\mu}-\beta m)^2] \nonumber                           \\
                                       & = \frac{1}{\beta^2} \boldsymbol{\mathrm{E}}[(\hat{\mu}-\mu + \mu - \beta m)^2] \nonumber                                                                        \\
                                       & =  \frac{1}{\beta^2}  \boldsymbol{\mathrm{E}}[(\hat{\mu}-\mu)^2] +  \frac{1}{\beta^2}  \boldsymbol{\mathrm{E}}[(\mu-\beta m)^2]   \nonumber                     \\
                                       & + \frac{2}{\beta^2} \boldsymbol{\mathrm{E}}[(\hat{\mu}-\mu)(\mu-\beta m)] \nonumber                                                                             \\
                                       & \approx  \frac{1}{\beta^2}  \boldsymbol{\mathrm{E}}[(\hat{\mu}-\mu)^2] + \frac{1}{\beta^2}  \boldsymbol{\mathrm{E}}[(\mu-\beta m)^2] \label{eq:samp_var_decomp}
  \end{align}

  In above equations, we approximate $\frac{2}{\beta^2}
    \boldsymbol{\mathrm{E}}[(\hat{\mu}-\mu)(\mu-\beta m)]$ to 0, since the two RVs
  $\hat{\mu}-\mu$ and $\mu-\beta m$ are approximately uncorrelated (since the
  randomness in the former RV comes from the random seeds in the hash function
  $h(\cdot)$, whereas that in the latter RV comes from another hash function
  $g(\cdot, \cdot)$), and the fact $\boldsymbol{\mathrm{E}}[\mu-\beta m] \approx
    0$.


  To prove Theorem~\ref{theo:var_samp}, it suffices to prove that the two terms
  in Eq.~\ref{eq:samp_var_decomp} are approximately equal to the corresponding
  terms in Theorem~\ref{theo:var_samp}. For the second term, this is obvious,
  since $\boldsymbol{\mathrm{E}}[(\mu-\beta m)^2] =
    \boldsymbol{\mathrm{Var}}[\mu] = m\beta(1-\beta)$ given that $\mu$ is
  distributed as a binomial RV $\mathcal{B}(m, \beta)$. For the first term,
  $\boldsymbol{\mathrm{E}}[(\hat{\mu}-\mu)^2]$ is precisely the (unconditional)
  mean squared error (MSE) of $\mu$, whereas the corresponding term in
  Theorem~\ref{theo:var_samp} is approximately the conditional variance
  $\boldsymbol{\mathrm{Var}}[\hat{\mu}|\mu=\beta m]$ (replacing $m$ in
  Theorem~\ref{the:var} by $\mu=\beta m$). To see the equivalence between these
  two terms, note that the MSE $\boldsymbol{\mathrm{E}}[(\hat{\mu}-\mu)^2]$ is
  equal to $\boldsymbol{\mathrm{Var}}[\hat{\mu}] + \mathcal{E}^2$ by definition,
  where $\mathcal{E}\triangleq \boldsymbol{\mathrm{E}}[\hat{\mu}-\mu]$ is the
  bias of the estimator $\hat{\mu}$. Since the bias $\mathcal{E}$ is small (using
  arguments almost identical to that in~\cite{whang} for a similar ``small bias
  proof''), we know that the MSE is approximately equal to the unconditional
  variance $\boldsymbol{\mathrm{Var}}[\hat{\mu}]$. Finally,
  $\boldsymbol{\mathrm{Var}}[\hat{\mu}]$ can be approximated by the conditional
  variance $\boldsymbol{\mathrm{Var}}[\hat{\mu}|\mu=\beta m]$, for the following
  two facts: (1) the value of $\boldsymbol{\mathrm{Var}}[\hat{\mu}] \approx
    \frac{n}{4}(e^{4\mu/n}-4\mu/n-1)$ (according to Theorem~\ref{the:var}), as a
  function of $\mu$, is slowing varying in the neighborhood around $\mu = \beta
    m$; and (2) the measure of the Binomial RV $\mu$ heavily concentrates around
  its mean $\beta m$.
\end{proof}

\noindent
\textbf{Numerical Analysis.}
The primary purpose of Theorem~\ref{theo:var_samp} is to guide the parameter tuning (say of the 
sampling rate $\beta$) in OddEEC.
Since Theorem~\ref{theo:var_samp} uses some approximations (e.g., $n$ is not large enough for the approximation in Proposition~\ref{prop:var} to be accurate), 
we would like to know if it is still accurate enough for this purpose.  To this end, 
Figure~\ref{fig:est_real} compares the estimated relative mean squared error (rMSE), computed using Theorem~\ref{theo:var_samp}, to the actual rMSE of OddEEC (using the method of moments estimator in Eq.~\ref{eq:estimator}). The results in Figure~\ref{fig:est_real} are based on 1000 experiments with $n=96$ and $\beta=0.5$. 
As shown in the figure, although Theorem~\ref{theo:var_samp} underestimates the error by 
 approximately 20\% to 50\% relative to the true rMSE, the overall trend of the estimated rMSE closely matches that of the real rMSE. In particular, Theorem~\ref{theo:var_samp} effectively predicts the bit error rate (BER) $\theta$ at which OddEEC achieves optimal performance.  As such, Figure~\ref{fig:est_real} demonstrates the efficacy of Theorem~\ref{theo:var_samp} for its primary purpose.

\begin{figure}[!htb]
  \begin{center}
    \includegraphics[width=0.9\columnwidth]{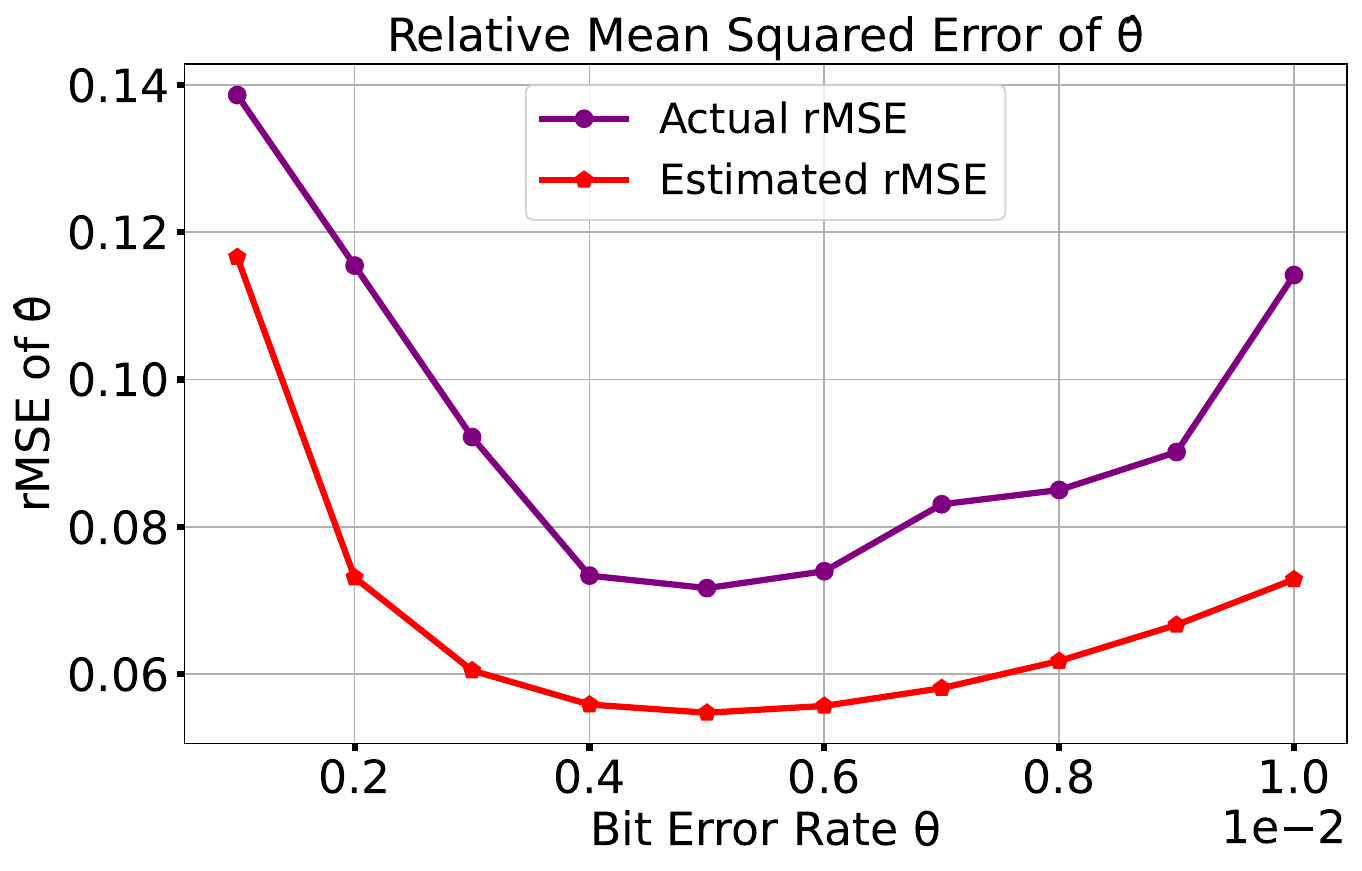}
    \caption{ Estimated rMSE versus actual rMSE when packet length $l=12,000$ bits for BER ranges from $0.001$ to $0.01$.}\label{fig:est_real}
  \end{center}
\end{figure}

\noindent
\textbf{Policy for Selecting Sampling Length.}
Apart from avoiding the ``saturation'' issue above, the bit sampling technique
also allows OddEEC to flexibly tune the sampling rate $\beta$ for near-optimal
estimation accuracy. For a fixed packet length $l$ and BER $\theta$ (and the
corresponding Hamming distance $m=\theta l$), guided by the variance formula in
Theorem~\ref{theo:var_samp}, we can select a near-optimal sampling rate $\beta$
as the one that minimizes the variance $\boldsymbol{\mathrm{Var}}[\hat{m}]$.

In reality, we cannot determine $\beta$ this way, since the receiver does not
know $\theta$ beforehand, even approximately. Hence, in OddEEC, the sampling
policy is oblivious of $\theta$; it is also extremely simple (so as not to
incur much computational overhead). Our sampling policy is to use a fixed {\it
    sampling length} $r$ (bits) in the following sense: Given a packet $P$ of
length $l$ (bits), bits in $P$ are sampled with rate $\beta = r/l$ if $l>r$ (so
roughly $r$ bits are sampled); sampling is not used if $l \le r$.




Now we explain how to select this $r$ (to be fixed for OddEEC operation). Given
the target ranges of $l$ (say $[l_{min}, l_{max}]$) and $\theta$ (say
$[\theta_{min}, \theta_{max}]$), the maximum number of ``balls'' (bit
errors) sampled by our policy is, {\it on average}, $\boldsymbol{\mathrm{E}}[\mu_{max}] = r\theta_{max}$ (assuming $l_{max}> r$).  We need to make sure $\boldsymbol{\mathrm{E}}[\mu_{max}]/n \le 2/3$ (called the safety
condition), since if $\mu/n\le 2/3$ (wherein $\mu$ is the actual number of ``balls'' sampled),
the probability (risk) for saturation (having $z\ge n/2$ as mentioned above)
would be acceptably low (no more than $0.006$).
Provided this safety condition is to be ensured, we tune $r$ to optimize the
accuracy of the estimator $\hat{m}$ (see Eq.~\ref{eq:estimator}), for the case
$\theta = (\theta_{min}+\theta_{max})/2$ (the actual BER is the midpoint of the
range) and $l = l_{max}$ (worst-case packet length). For example, when
$l_{max}=12,000$ and $\theta$ ranges from $0.001$ to $0.01$,
we optimize $\beta$
(correspondingly $r = \beta l$ here) assuming $m = (0.001 + 0.01) * 12,000/2 =
  66$, resulting in $r = 6000$ bits (so in this case, the effective $\beta$ is $6000/12,000 = 50\%$).
When $n = 96$ bits (OddEEC codeword length), $r = 6000$ ensures the safety
condition above (but barely as will be elaborated later), since in this case $\mu = 6000 \times 0.01 = 60 < 2/3 \times 96$. 

\input{sections/mle}

%% file: sections/mle.tex
\subsection{Maximum Likelihood Estimator}\label{subsec:mle}
In this subsection, we describe the maximum likelihood estimator (MLE) of OddEEC. We start with the description of MLE in one resolution and then describe how it can be extended to the multiple resolutions case.
In the following description, we remove the {\it with immunity} assumption made in
Section~\ref{subsec:oddeec}.  In other words, we now assume (realistically) that
the EEC codeword transmitted with a packet will also be subject to bit errors
(i.e., corrupted) with same BER $\theta$ (\emph{without immunity}). 
Unlike in previous sections, where we first derive $\hat{m}$ and then estimate
BER as $\hat{\theta} = \hat{m}/l$, here we directly estimate the BER $\theta$
by its MLE $\hat{\theta}$, which we derive as
follows.  Due to this difference, a slightly different probabilistic model is used here, and as a result
the probability $p$ derived in Eq.~\ref{eq:prob_odd_mle} differs slightly from that in Eq.~\ref{eq:prob_odd}. 


Recall that, in the case with immunity, we used $z$, the number of $1$-bits in
$S_{P} \oplus S_{P'}$, for estimating BER as described in~\ref{subsec:oddeec}.
However, in the case without immunity, the received EEC codeword, denoted as
$S^{c}_{P}$ (c for corruption) is in general not the same as the original
(uncorrupted) one $S_P$.
In this case, the receiver observes, instead of $z$, the number of 1-bits in
$S^{c}_{P} \oplus S_{P'}$, which we denote by $z^c$; here $S_{P'}$ is the
Odd sketch of the corrupted packet $P'$.
Suppose the observed value of $z^c$ is $\phi$. Then, the MLE $\hat{\theta}$ is
given by the following formula:
\begin{equation}\label{eq:mle}
  \hat{\theta} = \underset{\theta}{\mathrm{argmax}} \, Pr[z^{c}=\phi|\theta]
\end{equation}

The calculation of the conditional probability $Pr[z^{c}=\phi|\theta]$ (the
likelihood of $\theta$ given the observation $z^c$) breaks down into two parts. 
The first part ``reasons about'' $Pr[z=k | \theta]$, the conditional
probability of $z$ (that the receiver does not observe) taking a certain fixed
value $k$.
The second part ``reasons about'' $Pr[z^{c}=\phi | z=k,\theta]$, the
conditional probability of $z^{c}$ (that the receiver actually observes) of
taking the observed value $\phi$.
By the following total probability formula (that summarizes over all possible
values of $k$), we obtain
\begin{equation} \label{eq:total_prob}
  Pr[z^{c}=\phi|\theta] = \sum_{k=0}^{n} Pr[z=k | \theta] \cdot Pr[z^{c}=\phi|z=k,\theta]
\end{equation}

It remains to derive these two conditional probabilities (parts).
Since the first part corresponds to the case \emph{with immunity} (described
earlier in Section~\ref{subsec:oddeec}), $z$ follows a binomial distribution
$\mathcal{B}(n,p)$ (described 
in Section~\ref{subsec:samp}), where $n$ is the OddEEC codeword length (in number of bits)
and $p$ is given as follows:
\begin{align}
	p &=  \sum_{j=0}^{2j+1 \leq r} \binom{r}{2j+1} (\theta/n)^{2j+1}(1- \theta/n)^{r-2j-1}  \nonumber \\
	&= (1-(1-2\theta/n)^r)/2 \label{eq:prob_odd_mle}
\end{align}




Then the first conditonal probability in Eq.~\ref{eq:total_prob} is,
\begin{equation} \label{eq:prob_binomial}
  Pr[z=k|\theta] = \binom{n}{k} p^k (1-p)^{n - k}
\end{equation}

Now we derive the second part $Pr[z^{c}=\phi|z=k,\theta]$. The event
$\{z^{c}=\phi|z=k,\theta\}$ therein corresponds to that a bit string containing $k$
1-bits and $n-k$ 0-bits ends up containing $\phi$ 1-bits and $n-\phi$ 0-bits
after transiting through a channel with BER $\theta$.
The probability of this event has been well-studied
(e.g.,~\cite{information_theory}) and given by the following formula:
\begin{align}
  & Pr[z^{c} =\phi|z=k,\theta]  \nonumber \\
  & = \sum_{x} \binom{k}{x} \binom{n - k}{\phi - x} (1 - \theta)^{n - k-\phi+2x} \theta^{k+\phi-2x} \label{eq:prob_z}
\end{align}
This concludes the derivation of our MLE.

%% file: sections/multi_mle.tex
\subsection{Multi-resolution OddEEC}\label{subsec:multi_mle}

So far, our OddEEC scheme only has one resolution defined by the bit sampling length $r$ (described in the end of Section~\ref{subsec:samp}).  
However, as explained at the end of Section~\ref{subsec:samp}, a single resolution can only ``cover" (allows for accurate BER estimations in) a certain range $[\theta_{min}, \theta_{max}]$ due to the ``saturation issue";  and this range is often smaller than the 
range of BER values that of practical interest.  For example, as shown at the end of Section~\ref{subsec:samp}, if we use the optimized
sampling length $r=6000$ for a $12,000$-bit packet, (the MLE of) OddEEC gives accurate estimations only for the BER range $[0.001,0.01]$. 
As we will show in Section~\ref{sec:eval_eec}, when the actual BER is $0.05$, the BER estimation accuracy of this single-resolution OddEEC is
poor.


To deal with this coverage problem, we propose a multi-resolution OddEEC scheme.
In a multi-resolution scheme, the OddEEC codeword is partitioned into multiple subcodes that share the total (codeword length) budget, 
and a different sampling rate (length) is used for each subcode.  Since the sampling processes for the subcodes are independent, 
the joint likelihood function for a vector of $z^c$ observations is the product of their likelihood functions (derived in Eq.~\ref{eq:total_prob}).
As we will show in Section~\ref{subsec:comp_geec}, for the BER estimation range of $[0.001, 0.05]$ that we use in this paper, two resolutions provide 
a good ``coverage."  Hence, we will fix the maximum number of resolutions to 2 in the rest of this paper.

Now we show that our two-resolution OddEEC scheme has a very low decoding complexity (and more so for the single-resolution).
Denote the observed $z^c$ values in these two resolutions as $\phi_1$ and $\phi_2$ respectively.
The MLE decoding of OddEEC can be done in tens of nanoseconds using a precomputed MLE (value) table for all possible value 
combinations of $\phi_1$ and $\phi_2$, because the table is only a few thousand bytes in size, and hence can easily fit in main memory
or even cache.  For example, in the two-resolution scheme and setting used in our evaluation in Section~\ref{sec:eval_eec}, where each subcode has 48 bits (half of the 96-bit-long codeword), the table contains only $24 \times 24 = 576$ entries each of which is a 
4-byte-long (precomputed) MLE value.  It is $24 \times 24$ (rather than $48 \times 48$) here, because 
$\phi_1$ (or $\phi_2$) values larger than or equal to $24$ implies that accurate BER estimation is not possible from the first (or the second)
resolution due to the ``saturation effect" described in the beginning of Section~\ref{subsec:samp}.

Finally, other computations before the MLE decoding step, such as the computation (by the receiver) of the OddEEC codeword 
from the received packet $P$ and of the corresponding (observed) $\phi_1$ and $\phi_2$ values, can also be done in hundreds of nanoseconds,
using SIMD instructions such as \textit{\_mm\_and\_si128} and \textit{\_\_builtin\_popcount}~\cite{intel_set}.  
As such, the total decoding time is much shorter than the frame transmission time (of 1.25 microseconds) using the fastest WiFi technology as 
described at the end of Section~\ref{subsec:prob_stat}. 

In the following evaluation section, we only compare the running times of the MLE decoding step, because with the aforementioned SIMD instructions, the other computations take a similar amount of time for three EEC schemes (gEEC, mEEC, and OddEEC).

%% file: sections/evaluation.tex
\section{Evaluation of OddEEC}
\label{sec:eval_eec}
 In this section, we first evaluate the estimation accuracy of OddEEC through simulations and compare its performance with gEEC~\cite{geec}, the state-of-the-art method. 
 We do not compare against the original EEC since gEEC outperforms the original EEC as described  in Section~\ref{subsec:relate_eec}.
The results show that OddEEC outperforms gEEC within certain BER ranges. 
OddEEC also has significantly higher decoding (estimation) efficiency, using a small pre-computed MLE table as described in the second last paragraph in Section~\ref{subsec:multi_mle}.

Then, we compare OddEEC with mEEC~\cite{meec}.  For this comparison, we use a different evaluation setting:  
We assume that both the OddEEC and the mEEC codewords are immune from corruption (bit errors).
We use this ``with immunity" setting to our disadvantage (as OddEEC can deal with codeword corruption with ease), because
mEEC cannot deal with codeword corruption as stated in~\cite{meec}.
 Our results demonstrate that OddEEC achieves comparable estimation accuracy while requiring a significantly smaller precomputed table than mEEC.


\begin{figure*}[!htb]
  \begin{center}
    \includegraphics[width=0.9\columnwidth]{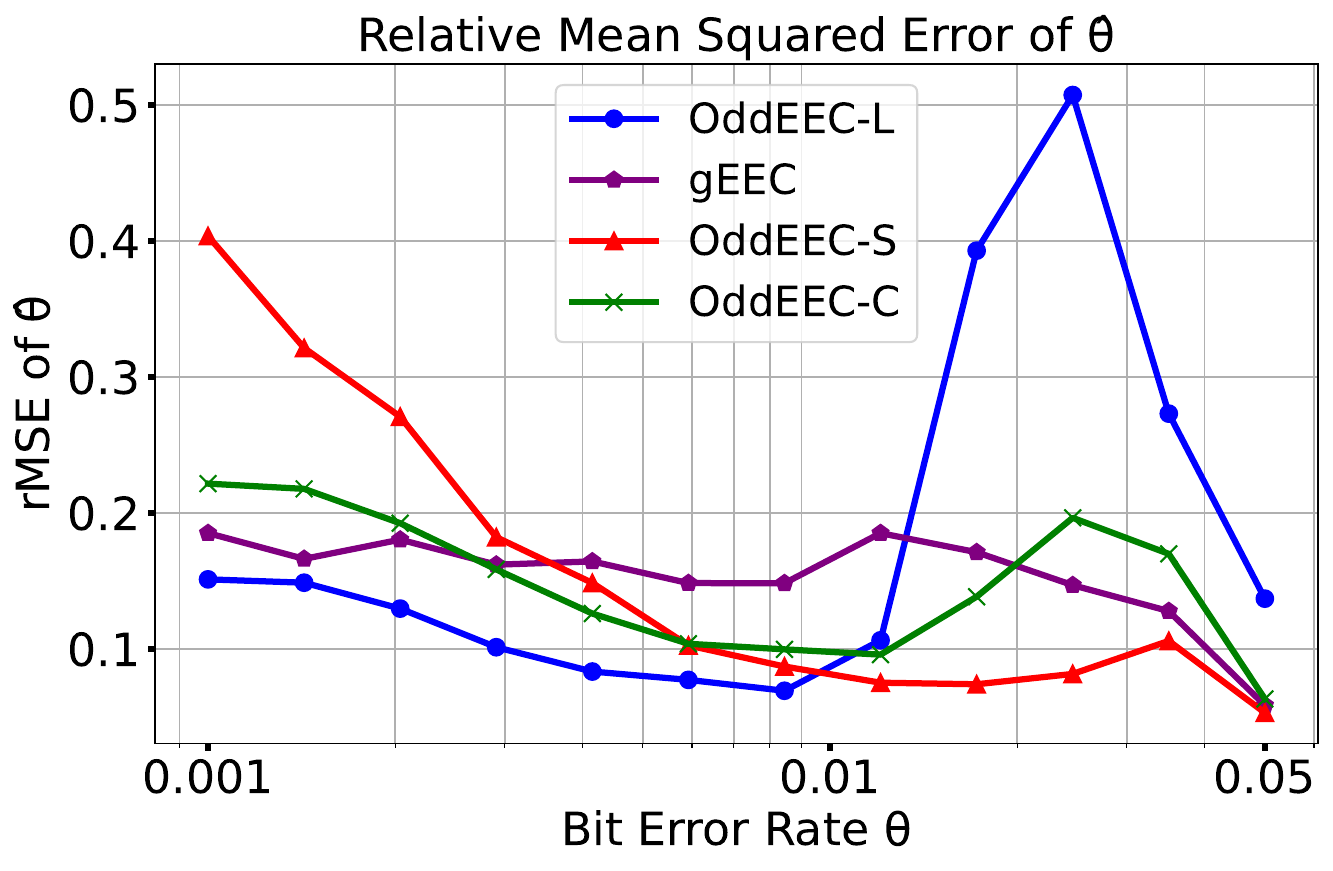}
    \includegraphics[width=0.9\columnwidth]{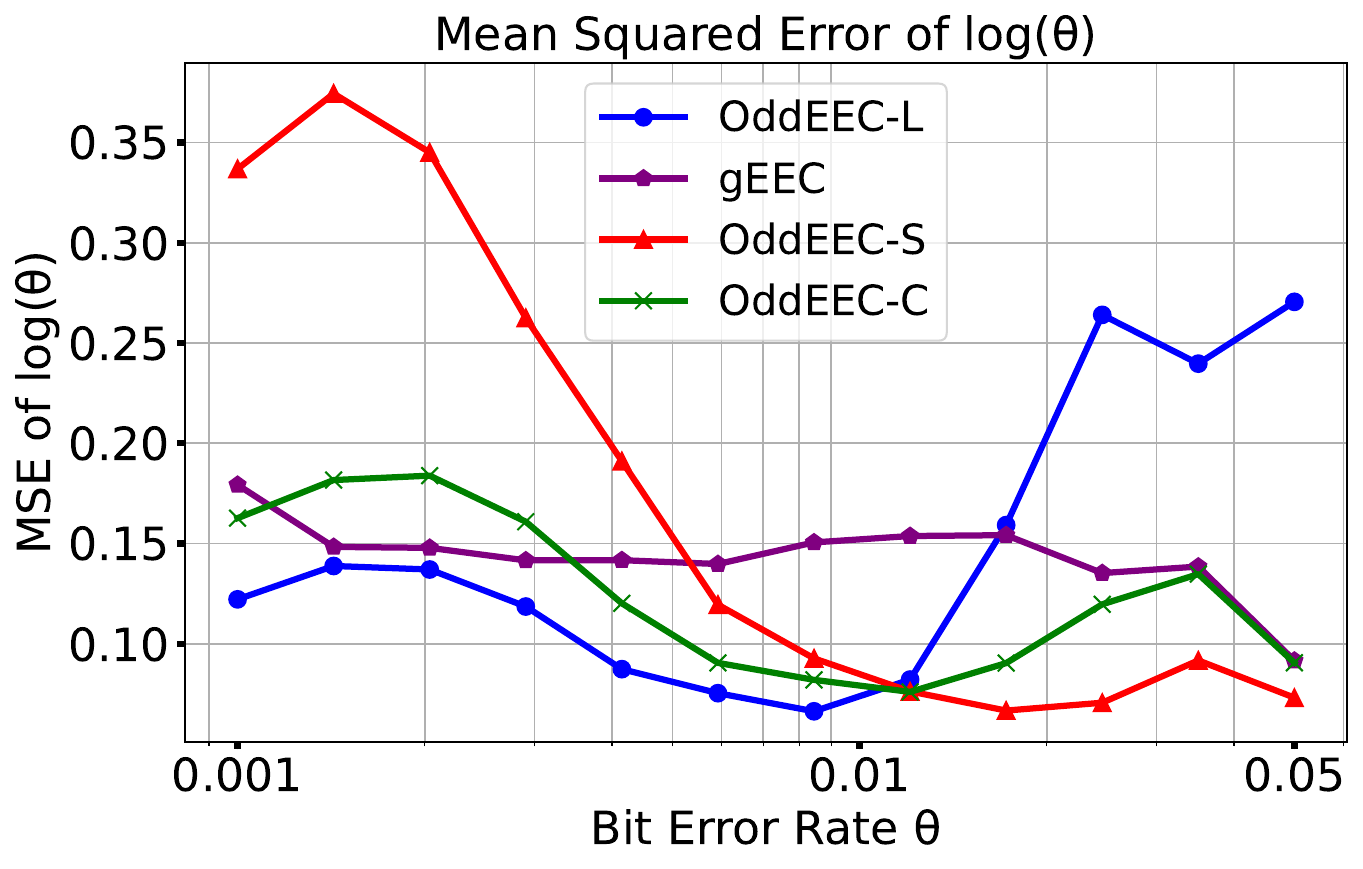}
    \caption{Accuracy comparison of OddEEC and gEEC when packet length $l=12,000$ bits.}\label{fig:eval_12,000}
  \end{center}
\end{figure*}

\begin{figure*}[!htb]
  \begin{center}
    \includegraphics[width=0.9\columnwidth]{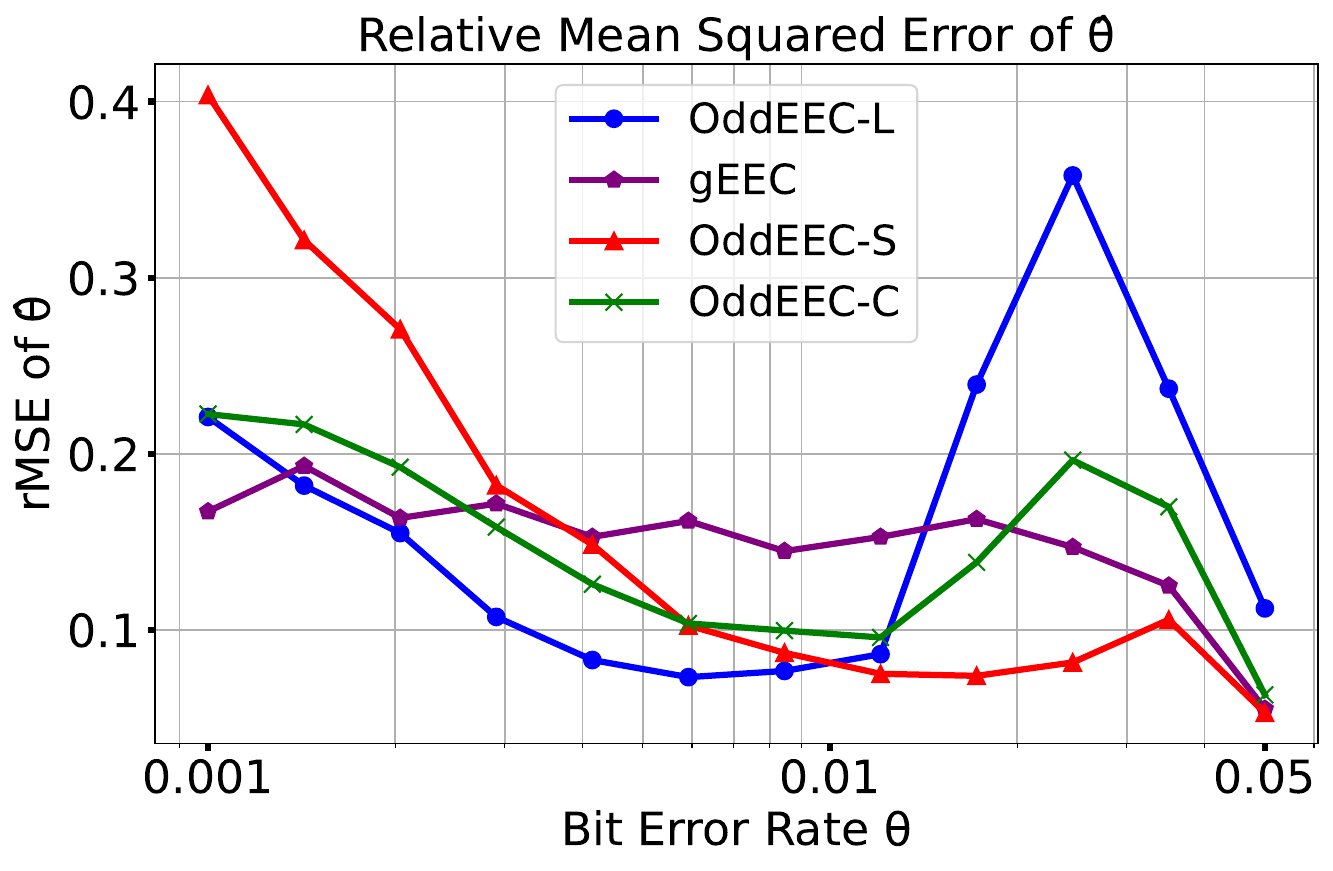}
    \includegraphics[width=0.9\columnwidth]{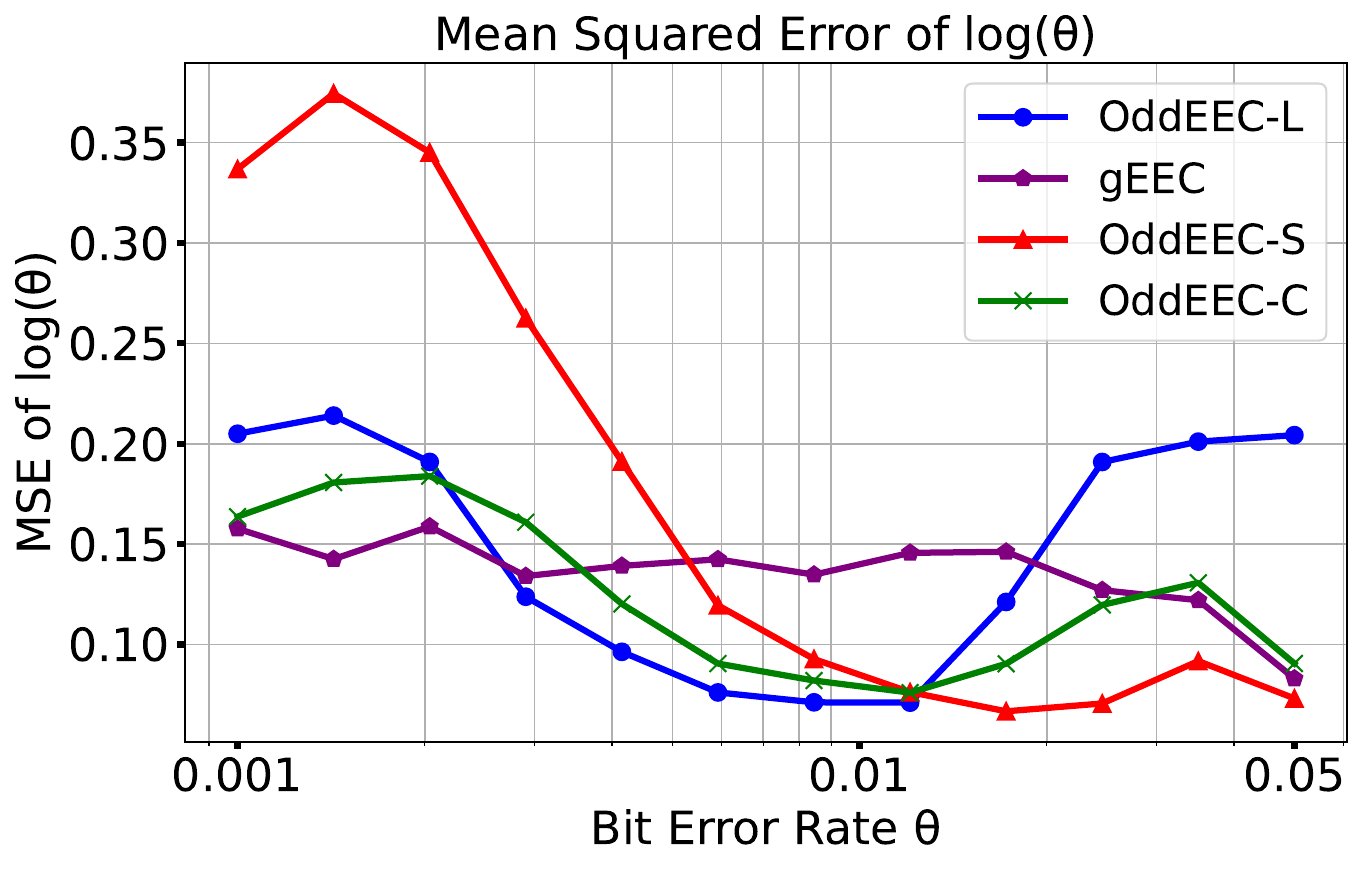}
    \caption{Accuracy comparison of OddEEC and gEEC when packet length $l=4000$ bits.}\label{fig:eval_4000}
  \end{center}
\end{figure*}

\vspace{0.4em}
\subsection{Evaluation Setup}\label{subsec:eva_setup}

In our evaluation, we use two packet lengths of $l=4000$ and $l=12,000$ bits
(which is both the typical and the maximum packet length in wireless network as described in~\cite{meec,geec,12000_usage}). 
Our target $\theta$ (BER) range is $[0.001, 0.05]$, since $0.05$ is considered a very high BER in wireless networking~\cite{low_snr,noma_ber} and hence a BER
higher than $0.05$ does not need to be accurately estimated (In this case it suffices to report that the BER is larger than $0.05$). Furthermore, as described in~\cite{sigcomm_ChenZZY10}, some applications like Wi-Fi rate adaptation only needs to know if the BER is close to or above a certain value and $0.05$ already appears to be a very large value in real-time network environment.  

We evaluate the accuracy of our MLE (for $\theta$) on the fourteen $\theta$ values geometrically even distributed in the range $[0.001,0.05]$.
For each combination of $l$ and
$\theta$, we generate 1000 random packets and introduce random bit errors (with rate $\theta$) to both these packets and their EEC codewords (i.e., without
immunity).  We report the average estimation accuracy (of the estimated BER $\hat{\theta}$) on these packets 
measured by the following two metrics like in
previous EEC papers~\cite{geec,meec}:
\begin{itemize}
  \item The Relative Mean Squared Error (rMSE), which is defined as:
        $\boldsymbol{\mathrm{E}}[{(\frac{\hat{\theta}-\theta}{\theta})}^2]$
  \item The Mean Squared log Error (logMSE), which is defined as:
        $\boldsymbol{\mathrm{E}}[(\log(\theta)-\log(\hat{\theta}))^2]$
\end{itemize}

In the interest of space, we do not present the accuracy comparisons according to two other metrics that were used in the literature~\cite{geec,meec}:
probability of large error; and relative bias.
Our comparisons according to these two metrics lead to similar conclusions.

 In all experiments, we fix the EEC codeword length $n$ to 96 bits for both gEEC and OddEEC. Specifically,
in gEEC, the codeword consists of $16$ sub-sketches of $6$ bits each, following
the same (optimal) configuration in~\cite{geec}. The gEEC codeword is computed
from 768 bits that are randomly sampled from the packet.

For OddEEC, we evaluate the following three configurations:
\begin{itemize}
\item We use a single resolution (that uses the entire 96-bit ``budget") with a sampling length of 
$r=2000$, and denote it as \textit{OddEEC-S} (where ``S" indicates a small sampling length).
\item We again use a single resolution (that uses the entire 96-bit ``budget"), but set the sampling length to 
$r=4500$.  We denote it as \textit{OddEEC-L} (where ``L" indicates a large sampling length).  Note that for a 4000-bit packet, since $r=4500>4000$, the packet is not sampled according to the sampling policy described in Section~\ref{subsec:samp}.
\item We employ a dual-resolution strategy in which each resolution uses 48 bits (so their total is 96), and the corresponding sampling lengths are $2250$ and $1000$
respectively (half of the sampling lengths used in the two one-resolution configurations). Since this configuration combines the sampling scales of OddEEC-S and OddEEC-L, we refer to it as \textit{OddEEC-C} (with `C' denoting ``combination").
\end{itemize}

To ensure fairness and improve estimation stability near high BER values, we upper-bound the MLE output at $0.06$ for all three OddEEC configurations, which is reasonable since it is slightly above $0.05$, the highest possible BER value that we intend to accurately measure.
For a fair comparison, this upper bound (of $0.06$) is also applied to gEEC and mEEC.  As a result of this upper-bounding, all evaluated algorithms exhibit improved accuracy as the BER approaches 
$0.05$.

\begin{figure*}[htb]
  \begin{center}
    \includegraphics[width=0.9\columnwidth]{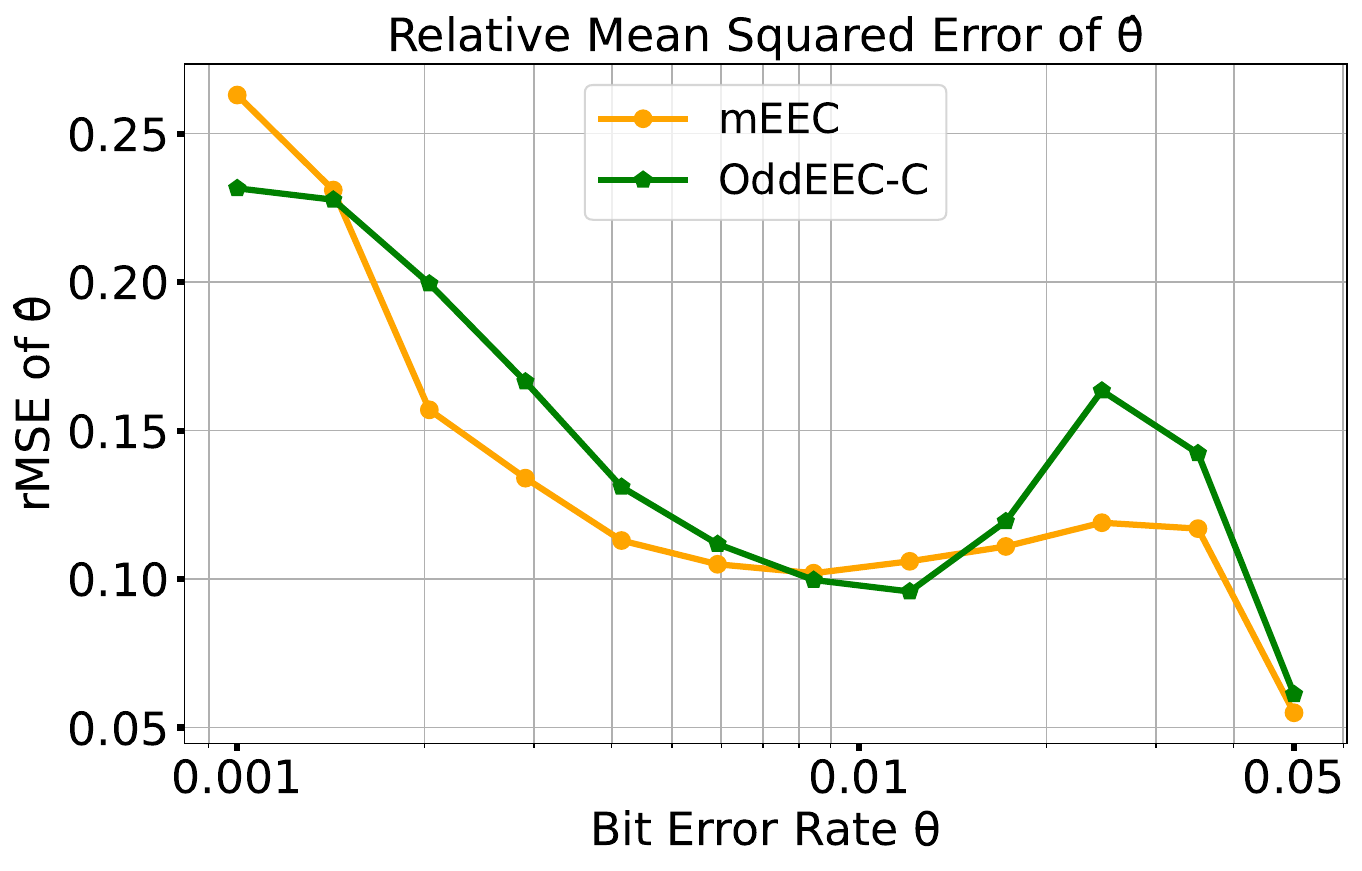}
    \includegraphics[width=0.9\columnwidth]{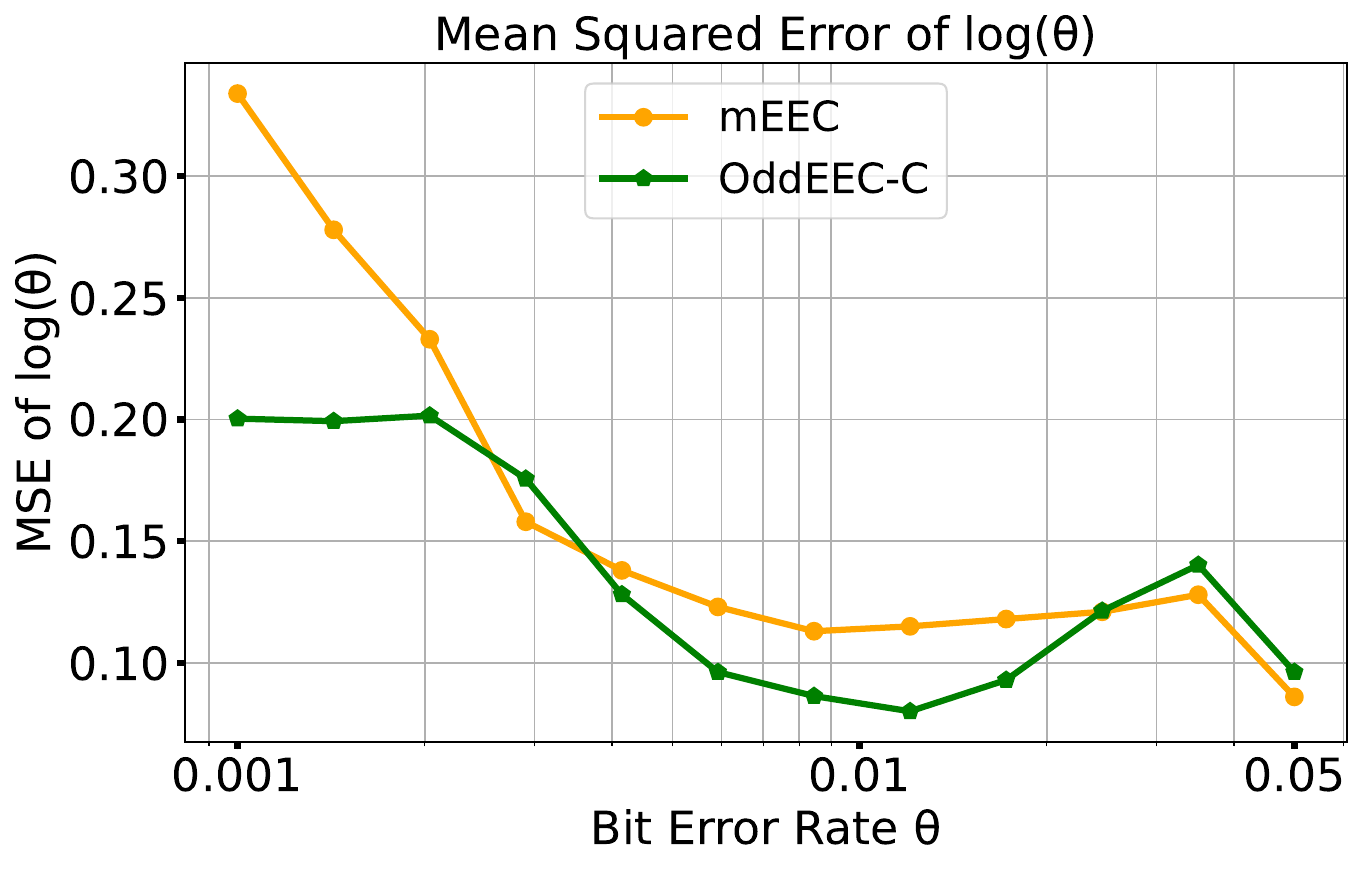}
    \caption{Accuracy comparison of OddEEC with mEEC when packet length $l=12,000$ bits.}\label{fig:eval_meec}
  \end{center}
\end{figure*}

\vspace{0.4em}
\subsection{Comparison with gEEC.}\label{subsec:comp_geec}
\noindent \textbf{Accuracy Comparison.}  Figures~\ref{fig:eval_12,000} and~\ref{fig:eval_4000} compare the four benchmark algorithms, on the rMSE (left subfigure) and the logMSE (right subfigure) metrics, 
for the 12,000-bit and the 4000-bit cases, respectively.  For each of the four algorithms except OddEEC-L, the corresponding curves are nearly identical in
(each subfigure of) Figure~\ref{fig:eval_12,000} and in (the corresponding subfigure of) Figure~\ref{fig:eval_4000}.
This is because all algorithms share the same parameter settings, except for OddEEC-L, which requires a different sampling length for 4000-bit and 12,000-bit packets, as explained at the end of Section~\ref{subsec:eva_setup}.  Nevertheless, the OddEEC-L curves exhibit similar shapes and trends, and hence lead to the same conclusion, across Figures~\ref{fig:eval_12,000} and~\ref{fig:eval_4000}.  Therefore, we focus our analysis on the results shown in Figures~\ref{fig:eval_12,000} (for 12,000-bit packet)
in the following.




As shown in Figure~\ref{fig:eval_12,000}, OddEEC-S achieves significantly higher accuracy than gEEC when the bit error rate (BER) is in the range 
$[0.004, 0.05]$:  OddEEC-S achieves a smaller rMSE (by factors between 1.11 and 2.46) and a smaller logMSE (by factors between 1.17 and 2.31). However, under low BER values ($\theta < 0.004$), OddEEC-S exhibits significantly worse accuracy than gEEC. This accuracy degradation is likely because the only input to OddEEC's estimator (Eq.~\ref{eq:mle}) is a single integer number $z^c$
(the number of discrepancies between two OddEEC codewords).  As the BER decreases,
the value of $z^c$ also becomes smaller, so the ``rounding error'' between the
observed $z^c$ value (always an integer) and its mean (which is in general a real number) 
becomes relatively (to $z^c$) larger.  In comparison, gEEC suffers much less from this ``rounding error'' issue,
since the input to its estimator consists of 16 integer values (readings from the 16 ToW sub-sketches). 

Increasing the sampling length mitigates this rounding error by increasing the expected Hamming distance between $P$ and $P'$ (after sampling), and hence 
the expected value of $z^c$.  As shown in Figure~\ref{fig:eval_12,000}, OddEEC-L achieves significantly higher accuracy than gEEC over the BER range $[0.001, 0.012]$, with rMSE being smaller by factors between 1.12 and 2.14 and logMSE being smaller by factors between 1.07 and 2.27.  However, at high BER values ($\theta > 0.017$), OddEEC-L’s accuracy starts to become worse than gEEC. This is because, when BER is larger than $0.017$, the expected Hamming distance $\mu$ is larger than $r*0.017 = 4500*0.017 = 77$, which exceeds aforementioned (at the beginning of Section~\ref{subsec:samp}) ``saturation threshold" $64 = 2/3*96$. 


By combining the power of two resolutions, OddEEC-C demonstrates comparable accuracy as gEEC across the evaluated BER range, as shown in Figure~\ref{fig:eval_12,000}. Specifically, it achieves a smaller rMSE by factors between 1.02 and 1.93 and a smaller logMSE by factors between 1.18 and 2.02 in the BER range $[0.004,0.02]$.  In the remaining BER range, OddEEC-C exhibits slightly worse estimation accuracy than gEEC, with rMSE being larger by factors between 1.07 and 1.33 and with logMSE being similar or slightly larger. The accuracy of OddEEC-C looks better in logMSE than in rMSE due to the 
following sensitivity profile of OddEEC.  When BER is large, OddEEC-C may still suffer from the saturation effect, occasionally outputting poor estimates (e.g., outputting the MLE upper bound of 0.06 even when the actual BER is much smaller). However, logMSE is less sensitive to such large deviations between the estimated and the true values than rMSE, making it a more favorable metric for OddEEC-C in high BER regions.

\noindent \textbf{Comparison of Decoding Efficiency.} 
Compared with gEEC, all three configurations of OddEEC offer significantly better decoding (estimation) efficiency, for the following reason. 
In gEEC, the MLE decoding process cannot be precomputed (as in OddEEC), because doing so would require precomputing and storing $2^{96}$ entries.
Hence, in gEEC, each MLE computation has to be performed at runtime for each received packet, which approximately takes tens of milliseconds.
This is several orders of magnitude longer than the decoding time (in tens of nanoseconds using a small lookup table, as explained at the end of Section~\ref{subsec:multi_mle}) needed by OddEEC. 
 
 To summarize, single-resolution OddEEC achieves much better accuracy than gEEC in some BER ranges while performing much worse in some other ranges. 
Two-resolution OddEEC, by combining the strengths of both resolutions, achieves slightly better accuracy than gEEC in some BER ranges while performing slightly worse in some other ranges.
OddEEC also has a significantly lower decoding (estimation) time, by approximately six orders of magnitude, than gEEC.

\subsection{Comparison with mEEC}
In this subsection, we compare OddEEC with mEEC.  In the interest of space, we present results only for packet length $l = 12,000$, as the comparison on 4000-bit packets leads to similar conclusions.  Additionally, we report only the results for the (two-resolution) OddEEC-C configuration, since it has been shown in the previous subsection that it overall outperforms both OddEEC-S and OddEEC-L across the entire evaluated BER range. The evaluation setup is identical to that used in the comparison with gEEC, except for the following two differences. First, as noted at the beginning of Section~\ref{sec:eval_eec}, mEEC is designed under the immunity assumption. To ensure a fair comparison, we modify OddEEC to operate under the same assumption by removing the second probability term (for handling the bit errors in the 
codeword) in Eq.~\ref{eq:mle}. Second, mEEC provides parameter settings only when the total code length is 80. Accordingly, we configure OddEEC-C to also use a total codeword length of 80 bits for consistency. Under this configuration, each subcode has a length of 40 bits, and the corresponding sampling lengths are set to 1900 and 840 bits, respectively.

As shown in Figure~\ref{fig:eval_meec}, OddEEC-C achieves overall comparable estimation accuracy as mEEC for 12,000-bit packets. Specifically,  OddEEC-C has a slightly higher rMSE than mEEC. 
However, in the logMSE metric, OddEEC-C achieves slightly better accuracy in the entire BER range except at BER values of 0.003 and 0.037. This is because logMSE is a more favorable metric for OddEEC-C, as explained in the third paragraph of Section~\ref{subsec:comp_geec}.  In terms of decoding complexity, both OddEEC and mEEC support precomputation, resulting in similar runtime efficiency. 
However, OddEEC has a significant advantage in the size of the precomputed lookup table.  Under the evaluated parameter settings, the precomputed table for OddEEC occupies only $4 \times 20 \times 20 = 1600$ bytes (similar to the reasoning at the end of Section~\ref{subsec:multi_mle}), whereas that for mEEC is approximately 20 megabytes~\cite{meec} in size, which is more than 12,000 times larger.  Moreover, mEEC relies on the immunity assumption, which is unrealistic in practical wireless environments. In contrast, OddEEC is designed to operate without this assumption, making it more applicable and robust for real-world deployment.

%% file: sections/conclusions.tex
\section{Conclusions}
\label{sec:conclusions}
In this paper, we introduce OddEEC, a novel Error Estimating Code (EEC) that
achieves better coding efficiency and significantly lower decoding complexity 
than the state of the art. We
develop a bit sampling technique and a maximum likelihood estimator that
address the new challenges of adapting Odd Sketch to an EEC scheme. We derive extensive analysis on the estimation accuracy of
OddEEC,  and demonstrate that
OddEEC achieves significantly lower decoding complexity than competing schemes while simultaneously achieving better overall estimation accuracy.

